\documentclass[11pt]{article}
\usepackage{fullpage}
\usepackage{amsmath}
\usepackage{amssymb}
\usepackage{amsthm}
\usepackage{epsfig}
\usepackage{pxfonts}
\usepackage{color}
\usepackage[numbers]{natbib}
\usepackage{hyperref}
\usepackage{tikz}
\usepackage{enumerate}
\usepackage{appendix}
\usepackage{soul}
\usepackage{float}
\usepackage{multirow}
\usepackage{graphicx}
\usepackage{caption}
\usepackage{subcaption}
\usepackage{enumitem}
\usepackage{filecontents}
\usepackage{amsmath, amsthm, amssymb}
\usepackage{amsthm}
\graphicspath{{./figures/}}

\newtheorem{theorem}{Theorem}[section]
\newtheorem{lemma}[theorem]{Lemma}
\newtheorem{proposition}[theorem]{Proposition}
\newtheorem{example}[theorem]{Example}
\usepackage{soul}
\makeatletter

\newcommand{\Rmnum}[1]{\expandafter\@slowromancap\romannumeral #1@}
\makeatother

\begin{document}

\title{Managing Relocation and Delay in Container Terminals with Flexible Service Policies}
\author{Setareh Borjian
       \thanks{ ORC and CEE, MIT.
    Email: \protect\url{sborjian@mit.edu}.}
\and
Vahideh H. Manshadi
    \thanks{ School of Management, Yale.
    Email: \protect\url{vahideh.manshadi@yale.edu}.}
\and
Cynthia Barnhart
    \thanks{ CEE, MIT.
    Email: \protect\url{cbarnhar@mit.edu}.}
\and
Patrick Jaillet
    \thanks{ EECS and ORC, MIT.
    Email: \protect\url{jaillet@mit.edu}.}
}
\date{}
\maketitle

\begin{abstract}
We introduce a new model and mathematical formulation for planning crane moves in the storage yard of container terminals. Our objective is to develop a tool that captures customer-centric elements, especially service time, and helps operators to manage costly relocation moves.
Our model incorporates several practical details and provides port operators with expanded capabilities including planning repositioning moves in off-peak hours, controlling wait times of each customer as well as total service time, optimizing the number of relocations and wait time jointly, 
and optimizing simultaneously the container stacking and retrieval process. We also study a class of flexible service policies which allow for out-of-order retrieval.
We show that under such flexible policies, we can decrease the number of relocations and retrieval delays without creating inequities.
\end{abstract}

\section{Introduction}
\label{sec:intro}

A container terminal is a facility where import/export/transshipment cargo containers arrive and are later distributed to nearby cities or loaded onto vessels for onward transportation. On the seaside, containers are unloaded from (or loaded to) vessels using huge berth cranes, and on the land side, they are stacked in (or retrieved from) a storage yard using smaller mobile cranes. \textit{Internal} trucks transfer containers from one point to another \textit{inside} the terminal, whereas \textit{external} trucks transfer the containers from the storage yard to \textit{outside} the terminal, or vice versa. 

In a typical container terminal, the storage yard is the area in which containers are stacked temporarily before they are loaded onto vessels or external trucks. The storage yard is usually divided into several bays. There are several columns in each bay, formed by stacking containers in tiers. The slot occupied by a container in a bay, can be specified by the tier and the column in which the container is stacked (see Figure \ref{Figure1}). In this paper, we focus on yard operations management in a terminal that handles import containers. 

Upon arrival for a retrieval, the external truck driver checks in at the terminal entrance gate with the information of the container to be picked up, and then drives to the bay where the container is stacked. The container is retrieved by a yard crane and loaded onto the truck. If the container is not at the top of its column and is covered, the \textit{blocking} containers must be relocated. These relocations are costly for port operators, as they are not charged to the customer. Moreover, the relocations slow the retrieval process, requiring external trucks to wait to receive their containers.

In the storage area, the number of relocation moves and the average waiting time of external trucks are among the key performance measures (\citet{ref17}). The main question, then, that yard managers face is as follows: given the estimated arrival and departure times of the containers (in a certain period of time), what is the stacking, relocating, and retrieval plan that will minimize the number of relocations while retrieving each container within a reasonable time? In this paper, we address this question by developing mathematical models and tools that capture several practical considerations and provide useful insights that help port operators manage yard crane moves. 

\begin{figure}[htb!]
\centering
{\includegraphics[width=4cm,height=4cm]{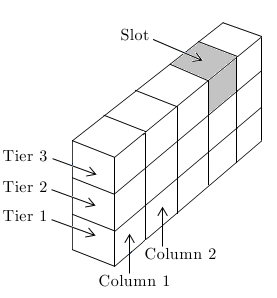}}
\caption{Illustration of a bay}
\label{Figure1}
\end{figure}

The above problem in a simplified setting is referred to as the Container Relocation Problem (CRP) or the Block Relocation Problem (BRP). The solution to the CRP is the sequence of moves to retrieve all containers with a minimum number of relocations, assuming that : \textbf{(1)} Containers are retrieved in a pre-defined order; \textbf{(2)} No container is stacked during the retrieval process; and \textbf{(3)} Only the containers that are on top of the target container are relocated. Even with this simplified setting, the problem is shown to be NP-hard by \citet{RePEc}, which explains why many proposed solution methods have been based on heuristics rather than exact optimization methods. Several authors consider heuristics based upon a set of rules designed to minimize the expected number of additional relocations as in the paper by \citet{ref5}, or to postpone (or avoid) future relocations as much as possible as presented in the paper by \citet{ref15} and in \cite{RePEc}.
Extending the objective function to the weighted sum of crane working time and number of relocations, \citet{ref4} present a three-phase heuristic that generates an initial feasible sequence of moves in the first phase and reduces the number of moves and crane's working time in phase two and three.
In other papers, authors use local search methods. \citet{ref10} use a tree search and \citet{ref12} use the corridor method to find the best slot for relocating a container. Alternatively, a number of papers are based on mathematical programming formulations for the CRP and attempt to solve them exactly. In one of the very first papers, \citet{ref5} present a formulation and provide a branch and bound algorithm for obtaining the optimal locations for relocated containers. In \cite{RePEc}, the authors identify the optimum sequence of relocation moves using a binary Integer Programming (IP) formulation. \citet{ref9} solve the same problem using a mixed IP formulation, and show that it has fewer decision variables and better runtime performance. \citet{ref19} and \citet{ref7} also present a binary IP for the CRP and then propose IP-based heuristics. For a classification scheme and a comprehensive literature review of the CRP and its variants in other practical applications, see the recent paper by \citet{Lit.Review}.

The critical assumption behind the CRP is that the containers are retrieved in a predefined \textit{order}, and the common aspect of all proposed mathematical models and heuristics is that they focus only on minimizing container relocations, and do not capture truck wait times or crane idle times when no truck is awaiting a retrieval. This can be significant in that such idle time allows for container \textit{repositioning} that will result in increased timeliness of service (reduced wait times) when trucks do arrive. Moreover, restricting the relocation only to containers above the target container can lead to suboptimal solutions. 

In our work, we introduce a new mathematical model that incorporates practical details, and captures customer-centric service elements, such as waiting times. In our model we include \textit{service times} (estimated departure and arrival times of containers) and determine the \textit{time} of each move as part of the solution.  Furthermore, we require that each retrieval and stacking task be completed within a \textit{service time window}, which can be set by the port operator based on task priorities and acceptable wait times. In addition, our model allows for taking into account idle time and planning repositioning moves that can be performed during such time. We refer to this generalized model as the \textit{time-based} CRP and present an Integer Programming (IP) formulation for it.

One of our main contributions is a new IP formulation that jointly minimizes the number of relocation moves (including repositioning moves) and service time, and provides the sequence of moves for retrieving containers for external trucks based on various service policies. For example, our model can be used by port operators to minimize external truck wait time, or to give different waiting time guarantees to different customers to reflect relative priorities.
Moreover, we expand the capabilities of existing approaches by determining repositioning moves that reduce overall service time, rather than the number of container relocations, thus moving from a focus of cost minimization for the port operators to one of maximization of customer service levels. 

In addition, we extend our model and formulation to the setting where containers are stacked and retrieved continuously, which we refer to as the Dynamic Container Relocation Problem (DCRP) \cite{ref7}. A special case of the DCRP is to find the best slots for stacking incoming containers such that they can be retrieved in the future with a minimum number of relocations. The problem has been addressed in a limited number of papers. \citet{ref16} use a heuristic to find the best stacking location based on the estimated departure time and the distance from the exit point. In another paper, \citet{ref19} apply index-based and IP-based heuristics to the DCRP, assuming a fixed stacking and retrieval order. In more recent work, \citet{ref7} propose an IP formulation and different heuristics that retrieve and stack containers in a pre-defined order. 
Also, some authors such as \citet{ref23} study variants of the DCRP by relaxing some of the constraints such as the columns height limit.

In our model, we relax the strict stacking ordering, and jointly optimize the ordering and the slot for stacking an incoming container. In the stacking process, we impose pre-defined orderings among \textit{groups} of containers (rather than for each container). The rationale for this is that typically, a group of containers is discharged from the vessel and is available for stacking with no difference in priority among the containers in the group.

We also propose and study a new class of retrieval policies. We relax the first-come-first-served (FCFS) policy and allow for \textit{out-of-order} retrievals. We show this flexible retrieval policy is surprisingly efficient and aligned with the benefits of port operators and customers. We demonstrate the impacts of this policy by measuring the number of relocations and service times. We prove and verify through experiments that such a policy results in fewer relocations and reduced service times. Moreover, our experiments show that flexible retrieval planning results in equitable customer services, an intuitively clear notion that which we will define in our context.

The structure of this paper is as follows. In Section \ref{formulation}, we present our time-based model and formulation and present an extended formulation for the DCRP. In Section \ref{sec 4}, we propose a flexible policy for retrieving containers and present the results of numerical experiments using this policy. Finally in Section \ref{discussion}, we summarize results, and suggest directions for future work.

\section{Mathematical Model and Formulation for the Time-Based CRP}
\label{formulation}
In this section, we present the mathematical model for the time-based CRP. We also develop an IP formulation that expands on the $BRP-I$ formulation presented in \cite{RePEc}. We illustrate cases that arise operationally and can be handled only by our time-based model. Moreover in Section \ref{DCRP}, we give an extension of the model that is capable of handling retrieval and stacking processes and thus can be used for solving the dynamic container relocation problem. The mathematical model and formulation in this section require that containers be retrieved in a pre-defined order (FCFS policy). In Section \ref{sec 4}, we relax this assumption and introduce a class of policies that allow for out-of-order retrievals. Table \ref{table:notations} summarizes the notations that will be used in the rest of the paper for describing the bay configuration and moves, defined as stacking, relocation, or retrieving one container.

In the time-based CRP, we are given a bay with $C$ columns and $P$ tiers, where containers are stacked on top of each other up to height $P$ (typically four or five) in each column. In a given time horizon (say a day), $N$ containers need to be retrieved from this bay and delivered to external trucks. The departure time of each container is the earliest time that it can be retrieved, and is the same as the arrival time of its corresponding truck. Throughout this paper, departure times are assumed to be known beforehand. 

\begin{table}[htb]
\begin{center}
\begin{tabular}{ l c  l}
 \hline
&  Notation & Description\\
\hline

Bay configuration & $ C\times P$  & A bay with C columns and P tiers.\\ \\

Container & $(c_n,d_n)$ & A container with label $c_n \in \{c_1, c_2, c_3,\dots\}$. The time-step  \\

& & at which the external truck corresponding with container $c_n$  \\\

& & arrives, is denoted by $d_n \in (1, 2, 3,\dots)$.  \\ \\


Slot & $[i,j]$ & The slot in the $i^{th}$ column and $j^{th}$ tier of the bay where\\
&& $1 \leq i \leq C$ and $1 \leq j \leq P$. Tiers are indexed from the bottom \\

&& such that the lowest and topmost tiers are represented\\
&& by $j=1$ and $j=P$, respectively. \\ \\

Relocation move & $[i,j] \rightarrow [k,l]$ & A relocation from slot $[i,j]$ to slot $[k,l]$.\\ \\

Retrieval move &  $[i,j] \rightarrow out$ & Retrieving a container from slot $[i,j]$.\\ \\

Schedule & $\textit{S} \in R^N$ & The arrival time-steps of external trucks that pick up \\
&& containers $c_1, c_2, \dots, c_N$.\\
\hline
   
\end{tabular}
\end{center}
\caption{Notations}
\label{table:notations}
\end{table}

Our time-based model relies on actual service times rather than the departure orders of containers. To define arrival times of external trucks for input to the model, we discretize the planning horizon into time-steps, where one time-step is the minimum time needed to complete one move. We then translate the arrival time of the trucks to time-steps. For example if the planning horizon is from $9 am$ to $11 am$ and each move takes on average $4$ minutes to complete, there will be $30$ time-steps in the planning horizon. In this example if an external truck arrives at 9:15, its arrival time-step is $4$ (note that if a truck arrives in the middle of one time-step, we assign it to the next time-step; so two or more external trucks may have the same arrival time, but we require that they be served in the order that they arrived).
For describing the bay configuration and moves, we use the same variables as those in $BRP-I$ \cite{RePEc}. However, unlike $BRP-I$ in which the time index, $\textit{t}$, is merely for determining the \emph{order} of the moves, in our model it represents the actual time of the day that each move is performed. Moreover, the input of the model is the arrival time of the external trucks rather than the departure order of containers. 
The four sets of variables are as follows:

\begin{eqnarray}
b_{ijnt} & = & \left\{
\begin{array}{l l}
1 & \quad \textrm{if container $n$ is in [$i$,$j$] at time $t$,}\\
0 & \quad \textrm{otherwise;}
\end{array} \right. \label{eq:var1} \\
x_{ijklnt} &=& \left\{
 \begin{array}{l l}
    1 & \quad \textrm{if container $n$ is relocated from [$i$,$j$] to [$k$,$l$] at time $t$,}\\
    0 & \quad \textrm{otherwise;}
  \end{array} \right.\label{eq:var2} \\
y_{ijnt} & = & \left\{
 \begin{array}{l l}
   1 & \quad \textrm{if container $n$ is retrieved from [$i$,$j$] at time  $t$,}\\
   0 & \quad \textrm{otherwise;}
 \end{array} \right.\label{eq:var4} \\
v_{nt} &=& \left\{
   \begin{array}{l l}
     1 & \quad \textrm{if container $n$ has been retrieved at time $t^{'}\in$\{1,...,$t$-1\}, }\\
     0 & \quad \textrm{otherwise;}
   \end{array} \right. \label{eq:var6}
\end{eqnarray}

Given the arrival time of the external trucks, we can specify the time-step at which each container is ready to be retrieved. We denote the departure time of container $c_n$ (i.e., the time-step that $c_n$ is ready to be retrieved) by $d_n$.  Further, we denote by $\delta_n$ the maximum amount of delay (measured in time-step) that is allowed in retrieving container $c_n$. Note that retrieval delay is the same as the wait time of external truck, and we will use these two terms interchangeably.
Given $d_n$ and $\delta_n$, we define a retrieval time window for each container as $[d_n , d_n+\delta_n]$, and require that containers be retrieved within these time windows. Also, if the departure time of a container is unknown or beyond the planning horizon, we allow it to stay in the bay arbitrarily long by setting its departure time to infinity.

Let $T$ be the total number of time-steps required to retrieve the containers. Because each retrieval task has a due time determined by its service time window, $T$ can be expressed as follows:

\begin{equation}
T = \max_{1 \leq n \leq N, ~d_n < \infty} \{d_{n}+\delta_n \}.\label{T formula}\\
\end{equation}

The set of constraints that ensures the configuration is feasible at each time-step and at most one move is performed at each time-step is given by \eqref{const1c}-\eqref{const7c}, which are similar to the constraints in $BRP-I$ \cite{RePEc}. We add Constraints \eqref{const5c}-\eqref{const13} to ensure the containers are retrieved in the order given by their indices in the bay, and within their service time windows.

Constraints \eqref{const1c} ensure that at each time-step, each container is either in the yard, or has been retrieved. Constraints \eqref{const2c} and \eqref{const3c} ensure that each slot is occupied by at most one container, and there are no empty slots between containers in a column. Constraints \eqref{const4c} require that at each time-step, at most one move takes place. Constraint \eqref{const6c} update the configuration of the yard at time-step $t$ based on the configuration and the move that took place in time-step $t-1$. Constraints \eqref{const7c} capture the relations between retrieval variables and their corresponding configuration variables.

Constraints \eqref{const5c} ensure that at most $n-1$ retrievals are performed before retrieving each container, $(c_n,d_n)$. (Proposition \ref{pro_flex} in Section \ref{sec 4} shows that satisfying Constraints \eqref{const5c} imply that containers are retrieved in prescribed order). 
Constraints \eqref{const12} ensure that each container is retrieved only after the arrival of its corresponding external truck. Constraints \eqref{const9} ensure that containers are retrieved within their allowable time windows, and Constraints \eqref{const13} require that containers are not retrieved at some time beyond their allowed time window.

\begin{eqnarray}
&& \sum\limits_{i=1}^{C}\sum\limits_{j=1}^{P}b_{ijnt}+v_{nt}= 1,\quad n=1,...,N,\quad t=1,...,T  \label{const1c}\\
&& \sum\limits_{n=1}^{N}b_{ijnt} \leq 1,\quad i=1,...,C,\quad j=1,...,P,\quad t=1,...,T \label{const2c}\\
&&\sum\limits_{n=1}^{N}b_{ijnt} \geq \sum\limits_{n=1}^{N}b_{ij+1nt},\quad i=1,...,C,\quad j=1,...,P-1,\quad t=1,...,T \label{const3c}\\
&&\sum\limits_{i,k=1}^{C}\sum\limits_{j,l=1}^{P}\sum\limits_{n=1}^{N} x_{ijklnt} + \sum\limits_{i=1}^{C}\sum\limits_{j=1}^{P}\sum\limits_{n=1}^{N}y_{ijnt} \leq 1,\quad t=1,...,T \label{const4c}\\
&& b_{ijnt}=b_{ijnt-1}-\sum\limits_{k=1}^{C}\sum\limits_{l=1}^{P} x_{ijklnt-1} + \sum\limits_{k=1}^{C}\sum\limits_{l=1}^{P} x_{klijnt-1} - y_{ijnt-1}\label{const6c} \\
&&\quad \quad \quad \quad \quad \quad \quad  \quad \quad \quad \quad \quad \quad i=1,...,C,\quad j=1,...,P,\quad n=1,...,N,\quad t=2,...,T \nonumber\\
&&v_{nt} =\sum\limits_{i=1}^{C}\sum\limits_{j=1}^{P}\sum\limits_{t^{'}=1}^{t-1}y_{ijnt^{'}},\quad n=1,...,N,\quad t=1,...,T \label{const7c}\\
&&(t-n)\sum\limits_{i=1}^{C}\sum\limits_{j=1}^{P}y_{ijnt} + \sum\limits_{i=1}^{C}\sum\limits_{j=1}^{P}\sum\limits_{n=1}^{N}\sum\limits_{t^\prime=1}^{t-1}y_{ijnt^\prime}\le t-1 
\label{const5c} \\
&&\quad \quad \quad \quad \quad \quad \quad  \quad \quad \quad \quad \quad \quad t=1,\dots,T,\quad n=1,\dots,N \quad \textrm{and} \quad d_1<d_2<\dots<d_N \nonumber \\
&&\sum\limits_{i=1}^{C}\sum\limits_{j=1}^{P}\sum\limits_{t=1}^{d_{n}-1}y_{ijnt}=0,\quad n=1,...,N \label{const12}\\
&&\sum\limits_{i=1}^{C}\sum\limits_{j=1}^{P}\sum\limits_{t=d_{n}}^{d_{n}+\delta_n}y_{ijnt}=1,\quad n=1,...,N \label{const9}\\
&&\sum\limits_{i=1}^{C}\sum\limits_{j=1}^{P}\sum\limits_{t=d_{n}+\delta_n+1}^{T}y_{ijnt}=0,\quad n=1,...,N \label{const13}
\end{eqnarray}

The objective is to jointly minimize the relocations and delay (difference between the time that a container is delivered to an external truck and the time that the truck arrives). We define the objective function as the weighted sum of these two efficiency metrics. The total number of relocation moves can be computed as $\sum_{i,k=1}^{C}\sum_{j,l=1}^{P}\sum_{n=1}^{N} \sum_{t = 1}^{T}x_{ijklnt}$.

Also note that the earlier we retrieve container $n$, the larger is the summation $\sum\limits_{t=d_n}^{T}v_{nt}$.
Thus we use the negative of this quantity as a measure of delay in retrieving containers. We also define parameters $w_{rel}$ and $w_{r}$ as the weight factors for relocations and retrieval delays, respectively. The weight factors are set by port operator based on the port policies and sensitivity of port operations to the number of relocations and total delay. The optimization model for the time-based CRP is as follows:

\begin{eqnarray}
\min \textrm{ }
w_{rel}\sum\limits_{i,k=1}^{C}\sum\limits_{j,l=1}^{P}\sum\limits_{n=1}^{N}\sum\limits_{t=1}^{T}x_{ijklnt}- w_{r}\sum\limits_{n=1}^{N}\sum\limits_{t=d_n}^{T}v_{nt} \label{objective}\\
s.t. \eqref{const1c} - \eqref{const13} \quad \textrm{and} \quad 
b, x , y, v \ge 0 \nonumber
\end{eqnarray}

One of the main features of the time-based model
is that the arrival time of external trucks can be specified. This is important because the trucks do not necessarily arrive uniformly over time (especially during off-peak hours), resulting in \textit{idle} times, i.e., the time that the yard crane is not busy with the retrieval process. In practice, during idle times, port operators may reposition the containers, in order to speed up the retrieval process in the future. 

Incorporating service times into the model, we can account for the idle times and plan repositioning moves based on truck arrival schedules and bay configuration. Note that the repositioning might not affect the total number of relocations, but it certainly affects the service delay for future arriving trucks, because some of the unavoidable relocations are performed beforehand (during the repositioning process). Thus, a reasonable way to plan repositioning is by considering its effect on the number of relocations and service delay (for future truck arrivals). Through jointly minimizing these two objectives, our time-based model can decide whether to do repositioning during the idle period, and which repositioning moves are most beneficial. The decision depends on the relative weights of the two objectives ($w_{rel}$ and $w_r$). If $w_{rel}$ is large compared to $w_r$, the optimal decision will be to avoid repositioning or perform very few repositioning moves. Conversely, if $w_{rel}$ is low compared to $w_{r}$, it might be optimal to perform some additional relocations (during the repositioning process) in order to further decrease service delay. In Appendix A, we illustrate this trade-off between relocation and service delay.

We investigated whether it is beneficial for the port (and customers) to schedule trucks in such a way that containers can be repositioned regularly. To examine this policy, we study the effect of two types of traffic (uniform and non-uniform) on service delay and on the number of relocations.

Consider a $C\times P$ bay with $N$ containers with a given initial configuration. Suppose that there are $N$ external trucks each picking up one container. We study two cases regarding the arrival schedule of the external trucks: 1. Uniform traffic: $N$ external trucks arrive at the yard uniformly and continuously over time, i.e., exactly one truck arrives at each time-step, and 2. Non-uniform traffic: $N$ external trucks arrive at the yard in two batches (two bursts of traffic) and there is a period of no-arrivals between the arrival of the two batches. We generate 500 random instances of a $4\times 6$ bay with 18 containers, and solve each instance twice: with uniform and non-uniform arrival times of the external trucks. In the case of non-uniform traffic, the arrival schedule is such that the first 9 containers arrive at times 1 to 9, and the second batch arrives at times 20 to 28; no truck arrives between time 9 and 20. We then compare the average delay per container (over all instances) between the two cases.

Note that the number of relocations required to retrieve the containers is the same in both cases. To study the impact on service delay, we set the objective function as the weighted sum of relocations and retrieval delay with equal weights. For each instance, the moves for retrieving the first batch of containers ($c_1$ to $c_9$) are the same in both cases. In 38 out of 500 instances, no repositioning is performed during the off-peak hours. In other instances, the service delay of containers in the second batch is decreased when there is non-uniform traffic because the containers are repositioned, i.e., some of the relocation moves can be performed during the period with no truck arrivals.

Figure \ref{boxplot} shows the distribution of the difference of average delay per container for the two cases (uniform traffic delay - non-uniform traffic delay). The red line in the box is the median and the blue dot is the mean of the distribution for the 500 instances. The average delay (over all 500 instances) per container is decreased by $17\%$ when traffic is non-uniform.

\begin{figure}[htb]
\centering
{\includegraphics[width=5cm,height=5cm]{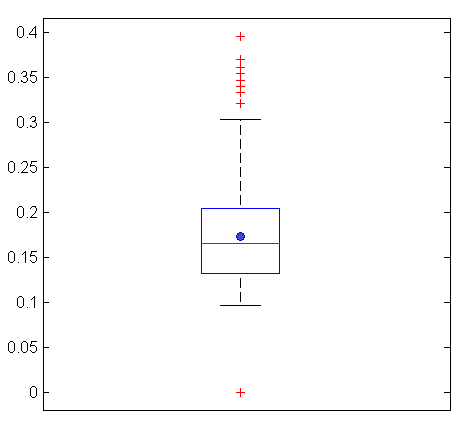}}
\caption{{\protect \footnotesize {Effect of traffic burst on delay per container: relative decrease in delay per container for the non-uniform traffic case compared to the uniform traffic case (average over 500 instances for $4\times 6$ bay)}}}
\label{boxplot}
\end{figure}

The impact of different scheduling policies on service delay is especially important in managing appointment systems and scheduling external trucks. In terminals with appointment system, truck drivers can make appointments to pick up or drop off containers during specific time windows and such a system has been shown to be very effective in reducing the number of relocation moves and service delay (\citet{Anne}). Given bay configuration and the current arrival schedule of external trucks, the appointment system can be used to recommend to the next truck a pick-up time that results in less retrieval delay and thus less waiting time.\\

The second important feature of the time-based model is the service-time window. Using time windows allows for more \textit{flexible} constraints and retrieval policies. For instance, it is possible for the port operator to choose time windows such that certain trucks (for example those registered in the appointment system in advance) experience less delay, or to serve some trucks (retrieve some containers) {\em out-of-order} with specified maximum amounts of service delay for each truck. 
This policy is discussed in detail in Section \ref{sec 4}. In addition to modelling flexibility, using time windows allows for more efficient computation by shrinking the search space of the IP (because we set many variables to zero by restricting the containers to be retrieved within a time window).

Note that one can set the time windows arbitrarily large without affecting optimality; i.e., any solution to a problem with a smaller time window is also a solution to the problem with larger time windows. However, for computational reasons, we strive to tighten the time windows as much as possible. In short, we are interested in the smallest time window that ensures problem  feasibility and optimality. Proposition \ref{time window UB} gives a bound with such properties for the retrieval time window of each container ($\delta^*_n$), for a container relocation problem with the objective to minimize the total number of relocations. 
We use the following lemma to prove Proposition \ref{time window UB}. 

\begin{lemma} 
Consider a CRP with the objective to minimize the total number of relocation moves. There exists an optimal sequence of moves, $M^*$, such that if at any time-step, $t$, there is a container that is ready to be retrieved, that time-step is not idle. 
\label{Lemma1}
\end{lemma}

\begin{proof}
See Appendix B.
\end{proof}

\begin{proposition}
\label{time window UB}
Consider a bay with $N$ containers, where $d_1<d_2<\dots<d_N$. Suppose the objective is to minimize the total number of relocation moves, $R^*$ is the optimal value when $\delta_n$ is arbitrarily large, and $H \ge R^*$ is an upper bound given by any heuristic (easily computable). We can reduce $\delta_n$ to $\delta^*_n$, given by \eqref{T.W:UB}, without increasing the number of relocations.

\begin{equation}
\delta^*_n = H + (n - d_n)^{+}.
\label{T.W:UB}
\end{equation}

\end{proposition}

\begin{proof}

In the time interval $[1,d_n+\delta_n]$, two types of moves (relocation and retrieval) or idle time-steps are possible. By definition, exactly $n$ containers must be retrieved by the end of time $d_n+\delta_n$. Moreover, $H$ provides an upper bound on the number of relocation moves ($H$ can be computed using existing fast heuristics).

To count the idle time-steps in time interval $[1,d_n+\delta_n]$, we use Lemma \ref{Lemma1} in the following way:
Suppose container $(c_n,d_n)$ is retrieved at time $\hat{d}_n \ge d_n$. The case that $\hat{d}_n = d_n$ is trivial because $\delta_{n}$ is zero, and any bound obtained for $\delta_n$ would be valid. 
If $\hat{d}_n > d_n$, then we know that in the time interval $[d_n,\hat{d}_n]$, container $c_n$ is available to be retrieved; so there is no idle time-steps in this interval. Also, in the time interval $[1,d_n-1]$ there are at least $n-1$ time-steps where a container is available to be retrieved, which means there are at most $d_n-1-(n-1)$ non-idle time-steps. Therefore, in the time interval $[1,d_n + \delta_n]$, there are at most $d_n-n$ idle time-steps (note that $\hat{d}_n = d_n + \delta_n$ because we have a solution in which $c_n$ is retrieved at $\hat{d}_n$). By counting the maximum number of relocations, retrievals, and the idle time-steps, we get an upper bound on the retrieval time window as follows:

\begin{equation}
d_n+\delta_n\le n+H+\max(d_n - n,0) \Rightarrow\delta_n \le H + (n - d_n)^{+}.
\end{equation}

Note that $\delta^*_n$ given by \eqref{T.W:UB} is the largest required time window. Also, it might be tight for some containers, but not for all.\\

\end{proof}

\subsection{The Dynamic container relocation problem}
\label{DCRP}
Up to this point we have assumed that all containers are already in the bay and we need only to retrieve them. In practice, the retrieval and stacking processes can overlap. In other words, at the same time that some containers are retrieved, some other containers need to be stacked in the bay; we refer to this problem as the dynamic container relocation problem (DCRP). It is not clear how to model the DCRP using retrieval and stacking order for the containers, because a single order cannot be specified for both tasks. Even defining two separate orders (for retrieval and stacking) does not work because first, it would not be clear which of the two tasks (retrieval or stacking) should be prioritized. Second, the stacking process is more flexible because typically a group of containers needs to be stacked and there is no priority for stacking one container before others. We can use this flexibility to reduce the number of relocations. In this Section, we extend the time-based model and formulation to the DCRP and  illustrate
how our model can find an optimal sequence of stacking and retrieval moves and take advantage of flexibility in the stacking process.

We define movement and configuration variables for the incoming containers in \eqref{eq:var3} and \eqref{eq:var5}, and add Constraints \eqref{const7} to capture the relationship between stacking variables and their corresponding configuration variables.

\begin{eqnarray}
s_{ijnt} & = & \left\{
    \begin{array}{l l}
      1 & \quad \textrm{if container $n$ is stacked in  [$i$,$j$] at time $t$,}\\
      0 & \quad \textrm{otherwise;}
    \end{array} \right.\label{eq:var3}\\
 z_{nt} &=& \left\{
   \begin{array}{l l}
     1 & \quad \textrm{if container $n$ has been stacked at time $t^{'}\in$\{1,...,$t$-1\}, }\\
     0 & \quad \textrm{otherwise;}
   \end{array} \right.\label{eq:var5} 
\end{eqnarray}

\begin{eqnarray}
&&z_{nt} =\sum\limits_{i=1}^{C}\sum\limits_{j=1}^{P}\sum\limits_{t^{'}=1}^{t-1}s_{ijnt^{'}},\quad n=1,...,N,\quad t=1,...,T. \label{const7}
\end{eqnarray}

We also replace Constraint \eqref{const1c} with Constraint \eqref{const1} to ensure that at each time-step, each container is either in the bay, outside the bay waiting to be stacked, or has been retrieved and delivered to an external truck.

\begin{eqnarray}
&& \sum\limits_{i=1}^{C}\sum\limits_{j=1}^{P}b_{ijnt}+v_{nt}=z_{nt},\quad n=1,...,N,\quad t=1,...,T.  \label{const1}
\end{eqnarray}

Moreover, we define $a_n$ as the arrival time of incoming container $n$ (the earliest time that it is ready to be stacked),  and $\alpha_n$ as the maximum amount of delay that is allowed in stacking container $n$, similar to $d_n$ and $\delta_n$ respectively for the arrival time of external trucks and retrieval time window. Note that the maximum number of required time-steps ($T$) is now given by \ref{T formula2}:

\begin{equation}
T = \max_{1 \leq n \leq N,~ 1 \leq m \leq N,~d_m < \infty} \{a_{n}+\alpha_n, d_{m}+\delta_m \}.\label{T formula2}\\
\end{equation}

We also add Constraints \eqref{const10} to \eqref{const11}. Constraints \eqref{const10} ensure that each container is stacked only after its scheduled arrival time. Constraints \eqref{const8} ensure that containers are stacked within their allowable time windows, and Constraints \eqref{const11} require that containers are not stacked at some time beyond their allowed time window. Moreover, we add stacking moves for updating bay configuration and replace Constraints \eqref{const6c} with Constraints \eqref{const5}. Lastly, Constraints \eqref{const4c} are replaced with Constraints \eqref{const4} to ensure that, at each time-step, at most one move (retrieval, staking, or relocation) is performed.

\begin{eqnarray}
&&\sum\limits_{i=1}^{C}\sum\limits_{j=1}^{P}\sum\limits_{t=1}^{a_{n}-1}s_{ijnt}=0,\quad n=1,...,N \label{const10}\\
&&\sum\limits_{i=1}^{C}\sum\limits_{j=1}^{P}\sum\limits_{t=a_{n}}^{a_{n}+\alpha_n}s_{ijnt}=1,\quad n=1,...,N \label{const8}\\
&&\sum\limits_{i=1}^{C}\sum\limits_{j=1}^{P}\sum\limits_{t=a_{n}+\alpha_n+1}^{T}s_{ijnt}=0,\quad n=1,...,N \label{const11}\\
&& b_{ijnt}=b_{ijnt-1}-\sum\limits_{k=1}^{C}\sum\limits_{l=1}^{P} x_{ijklnt-1} + \sum\limits_{k=1}^{C}\sum\limits_{l=1}^{P} x_{klijnt-1} - y_{ijnt-1}+s_{ijnt-1} \label{const5} \\
&&\quad \quad \quad \quad \quad  \quad \quad \quad \quad \quad \quad \quad \quad \quad i=1,...,C,\quad j=1,...,P,\quad    n=1,...,N,\quad t=2,...,T \nonumber\\
&&\sum\limits_{i,k=1}^{C}\sum\limits_{j,l=1}^{P}\sum\limits_{n=1}^{N} x_{ijklnt} + \sum\limits_{i=1}^{C}\sum\limits_{j=1}^{P}\sum\limits_{n=1}^{N}y_{ijnt} + \sum\limits_{i=1}^{C}\sum\limits_{j=1}^{P}\sum\limits_{n=1}^{N}s_{ijnt} \leq 1,\quad t=1,...,T. \label{const4}
\end{eqnarray}

Note that we do not need to define any order for stacking containers or sequencing stacking and retrieval moves; instead, we model the stacking and retrieval processes by specifying departure and arrival times, and time windows. Retrieval moves are performed in the predefined order (due to Constraints \eqref{const5c}). It is possible to enforce stacking moves to follow a predefined ordering as well; however, in the stacking process usually a batch of containers (as opposed to a single container) is discharged from a vessel and all are available to be stacked with no difference in their order. This is an important modelling issue that cannot be captured by specifying stacking order. In our proposed time-based model, we assign the same stacking time to all containers available for stacking at the same time and add Constraints \eqref{stack order} to respect the stacking order of groups of containers. The actual order in which the containers within a group are stacked is enforced by minimizing the number of relocations and total delay for the stacking and retrieval processes:

\begin{eqnarray}
\sum\limits_{t=1}^{T}z_{nt} \geq \sum\limits_{t=1}^{T}z_{n^{\prime}t},\quad \forall 
n, n^{\prime} \quad s.t. \quad a_{n^{\prime}} = \min\{a_i|a_i > a_n\}. \label{stack order}
\end{eqnarray}

For the DCRP, we add stacking delay and its corresponding weight to the objective function, the general form of which is given as follows:

\begin{eqnarray}
w_{rel}\sum\limits_{i,k=1}^{C}\sum\limits_{j,l=1}^{P}\sum\limits_{n=1}^{N}\sum\limits_{t=1}^{T}x_{ijklnt}- w_{s}\sum\limits_{n=1}^{N}\sum\limits_{t=a_n}^{T}z_{nt} - w_{r}\sum\limits_{n=1}^{N}\sum\limits_{t=d_n}^{T}v_{nt}.\label{general objective}
\end{eqnarray}

In Appendix C, we illustrate a case, applying the above model to find the best slots for stacking containers in a $3\times4$ bay such that the sum of total delay and number of relocations is minimized. Complete formulation of the DCRP is provided in Appendix D.\\

Throughout this section, our model and formulation assumed a first-come-first-served policy (similar to existing models in the literature and typical practices in container terminals). In the next section, we relax this assumption and modify the formulation to allow for out-of-order retrievals. We show that this retrieval policy helps reduce the number of relocations and retrieval delay, and therefore benefits both port operators and customers.
\section{Flexible Retrieval Planning}
\label{sec 4}

The formulation presented in Section \ref{formulation} requires that containers be retrieved in a pre-defined order, a first-come-first-served policy, which is enforced by Constraints \eqref{const5c}. 
In this section, we examine the FCFS in the retrieval process and propose a more general class of service policies which allows for flexibility in the order of retrievals.

Suppose that there are a number of trucks in the storage yard waiting in a queue for containers to be retrieved and delivered to them. If the trucks are served based on a FCFS policy, the containers have to be retrieved in the order dictated by the arrival order of the trucks. As an alternative to FCFS, we introduce the following flexible retrieval policy:
Suppose a container that needs to be relocated during the retrieval process (because it is blocking a target container) is ready to depart, i.e., its corresponding external truck has arrived and is waiting somewhere in the queue. An alternative to relocating that container is to retrieve it \emph{out-of-order}. We refer to this as a \textit{flexible retrieval planning} policy. To avoid inequity and customer dissatisfaction, the level of flexibility should be defined and controlled to limit the number of out-of-order retrievals.

We define the level of flexibility, $m$, as follows: for any container, $(c_n,d_n)$, at most $m$ containers whose departure time is greater than $d_n$, can be retrieved before $c_n$. 
For example if $m=1$, before retrieving container $(c_n,d_n)$, at most one container whose departure time is greater than $d_n$ may be retrieved. By definition, the FCFS policy is equivalent to $m=0$.

We add the following constraints to allow for out-of-order retrievals and set the level of flexibility $m$ in the model (all other constraints remain the same as before.). For a bay with $N$ containers and $d_1<d_2<\dots<d_N$, Constraints \eqref{flexible ret} ensure that at most $n-1+m$ retrievals are performed before retrieving container $(d_n,c_n)$. The following proposition shows that if Constraints \eqref{flexible ret} are satisfied, each external truck experiences at most $m$ out-of-order retrievals. 

\begin{eqnarray}
(t-m-n)\sum\limits_{i=1}^{C}\sum\limits_{j=1}^{P}y_{ijnt} + \sum\limits_{i=1}^{C}\sum\limits_{j=1}^{P}\sum\limits_{n=1}^{N}\sum\limits_{t^\prime=1}^{t-1}y_{ijnt^\prime}\le t-1, \nonumber \\
\quad t=1,\dots,T,\quad n=1,\dots,N \quad \textrm{and} \quad d_1<d_2<\dots<d_N. 
\label{flexible ret}
\end{eqnarray}

\begin{proposition}
\label{pro_flex}
For a CRP with $N$ containers, suppose that $d_1 < d_2 < \dots < d_N$. If before each container, $(c_n,d_n)$, at most $(n-1)+m$ containers are retrieved, then the number of out-of-order retrievals before each container is at most $m$.
\label{m out}
\end{proposition}

\begin{proof}
See Appendix C.
\end{proof}

Example \ref{ex.flex.ret} below compares the number of relocations, retrieval delay for each container, and total delay, for the two cases (with $m=0$ and $m=1$).

\begin{example} \label{ex.flex.ret}
\normalfont
\textbf{Out-of-order retrieval:}
Consider a $3\times4$ bay with 9 containers with the initial configuration given in Figure \ref{fig:example.flex.ret}. The arrival schedule of the external trucks that pick up the containers is given by $S=\{1,2,3,4,5,6,7,8,9\}$.

\begin{figure}[htb]
\centering
\begin{tabular}{ |c| c| c|}
\hline
& & \\
\hline
$(c_6,6)$&$(c_3,3)$&$(c_4,4)$\\
\hline
$(c_5,5)$&$(c_2,2)$&$(c_1,1)$\\
\hline
$(c_8,8)$&$(c_7,7)$&$(c_9,9)$\\
\hline
\end{tabular}
\caption{Initial configuration of the bay for Example \ref{ex.flex.ret}}
\label{fig:example.flex.ret}
\end{figure}

We solve this instance for two cases of $m=0$ and $m=1$ and compare the sequence of moves as shown in Figure \ref{ex.flex.ret.solution0} and \ref{ex.flex.ret.solution1}. For simplicity, the departure times of containers are not shown in the sequence of moves, and the containers are indicated by their indices.

\begin{figure}[htb!]
\begin{subfigure}[t]{1\textwidth}
\centering
\begin{tabular}{ |c| c| c| c| c | c| c| c| c| c| c |c| c| c| c| c | c| c| c| c| c| c| c| c}
\cline{1-3}\cline{5-7}\cline{9-11}\cline{13-15}\cline{17-19}\cline{21-23}
& & & & 4 & & & & 4 & & & & 4 & & & & 4 & & & & \rlap{/}{4} & & &\\
\cline{1-3}\cline{5-7}\cline{9-11}\cline{13-15}\cline{17-19}\cline{21-23}
$6$ &$3$ &$4$ &  &6 &3 & & & 6& 3& &  & 6& & &  & 6& & &  &6& & &   \\
\cline{1-3}\cline{5-7}\cline{9-11}\cline{13-15}\cline{17-19}\cline{21-23}

5& 2& 1& $\rightarrow$ & 5& 2& \rlap{/}{1}&$\rightarrow$ & 5& 2& & $\rightarrow$& 5& \rlap{/}{2}& 3&$\rightarrow$
 & 5& & \rlap{/}{3}&$\rightarrow$ & 5& & &   \\
\cline{1-3}\cline{5-7}\cline{9-11}\cline{13-15}\cline{17-19}\cline{21-23}
8& 7& 9& & 8& 7& 9& & 8& 7& 9& & 8& 7& 9& & 8& 7& 9& & 8& 7& 9& \\
\cline{1-3}\cline{5-7}\cline{9-11}\cline{13-15}\cline{17-19}\cline{21-23}
\cline{1-3}\cline{5-7}\cline{9-11}\cline{13-15}\cline{17-19}\cline{21-23}
\multicolumn{3}{c}{$t=1$} & \multicolumn{1}{c}{ } & \multicolumn{3}{c}{$t=2$} & \multicolumn{1}{c}{} & \multicolumn{3}{c}{$t=3$} & \multicolumn{1}{c}{ }& \multicolumn{3}{c}{$t=4$} & \multicolumn{1}{c}{} & \multicolumn{3}{c}{$t=5$} & \multicolumn{1}{c}{ }& \multicolumn{3}{c}{$t=6$}& \\
\multicolumn{20}{c}{ } & \multicolumn{3}{c}{$\downarrow$}& \\
\end{tabular}

\begin{tabular}{ |c| c| c| c| c | c| c| c| c| c| c |c| c| c| c| c | c| c| c| c| c| c| c| c }

\cline{1-3}\cline{5-7}\cline{9-11}\cline{13-15}\cline{17-19}\cline{21-23}
\textcolor{white}{1}&\textcolor{white}{1}&\textcolor{white}{1}&&\textcolor{white}{1}&\textcolor{white}{1}&
\textcolor{white}{1}&\textcolor{white}{1}&& & & & &  & &   & & & &  & & & &   \\

\cline{1-3}\cline{5-7}\cline{9-11}\cline{13-15}\cline{17-19}\cline{21-23}
 && & &&&&&&&&{\textcolor{white}{1}} & \textcolor{white}{1} & \textcolor{white}{1} & 
& \textcolor{white}{1}  & & & &  & 6 &  & &   \\

\cline{1-3}\cline{5-7}\cline{9-11}\cline{13-15}\cline{17-19}\cline{21-23}
&&& $\leftarrow$  &  & && $\leftarrow$& &&& $\leftarrow$ & &\rlap{/}{6}& &$\leftarrow$ & \rlap{/}{5}&6&& $\leftarrow$&5&&& \\

\cline{1-3}\cline{5-7}\cline{9-11}\cline{13-15}\cline{17-19}\cline{21-23}

 &&\rlap{/}{9}& & \rlap{/}{8}  &  & 9 & 
& 8  &\rlap{/}{7} & 9& & 8 & 7 &  9 & & 8& 7& 9 & & 8& 7&9 &\\

\cline{1-3}\cline{5-7}\cline{9-11}\cline{13-15}\cline{17-19}\cline{21-23}
\multicolumn{3}{c}{$t=12$} & \multicolumn{1}{c}{} & \multicolumn{3}{c}{$t=11$} & 
\multicolumn{1}{c}{ }& \multicolumn{3}{c}{$t=10$} & \multicolumn{1}{c}{} & \multicolumn{3}{c}{$t=9$} &
\multicolumn{1}{c}{ }& \multicolumn{3}{c}{$t=8$} &\multicolumn{1}{c}{ }& \multicolumn{3}{c}{$t=7$} &\\
\end{tabular}
\caption{Sequence of moves for the bay of Example \ref{ex.flex.ret}: inflexible case ($m=0$, FCFS)}
\label{ex.flex.ret.solution0}
\end{subfigure}

\begin{subfigure}[t]{1\textwidth}
\centering
\begin{tabular}{ c c c}
&&\\
\end{tabular}
\\
\begin{tabular}{ |c| c| c| c| c | c| c| c| c| c| c |c| c| c| c| c | c| c| c| c| c| c| c|}
\cline{1-3}\cline{5-7}\cline{9-11}\cline{13-15}\cline{17-19}\cline{21-23}
& & & & 4 & & & & 4 & & & & 4 & & & & \rlap{/}{4} & & & &  & & \\
\cline{1-3}\cline{5-7}\cline{9-11}\cline{13-15}\cline{17-19}\cline{21-23}
6&3&4 & & 6&3& & & 6&\rlap{/}{3}& & & 6&& & & 6&& & & \rlap{/}{6}&&  \\
\cline{1-3}\cline{5-7}\cline{9-11}\cline{13-15}\cline{17-19}\cline{21-23}
5&2&1& $\rightarrow$ & 5&2&\rlap{/}{1}& $\rightarrow$ & 5&2&& $\rightarrow$ & 5&\rlap{/}{2}&& $\rightarrow$ & 5&&& $\rightarrow$ & 5&& \\
\cline{1-3}\cline{5-7}\cline{9-11}\cline{13-15}\cline{17-19}\cline{21-23}
8&7&9 & & 8&7&9 & & 8&7&9 & & 8&7&9 & & 8&7&9 & & 8&7&9  \\
\cline{1-3}\cline{5-7}\cline{9-11}\cline{13-15}\cline{17-19}\cline{21-23}
\multicolumn{3}{c}{$t=1$}& \multicolumn{1}{c}{} & \multicolumn{3}{c}{$t=2$} & \multicolumn{1}{c}{} & \multicolumn{3}{c}{$t=3$} & \multicolumn{1}{c}{ }& \multicolumn{3}{c}{$t=4$} & \multicolumn{1}{c}{} & \multicolumn{3}{c}{$t=5$} & \multicolumn{1}{c}{ }& \multicolumn{3}{c}{$t=6$} \\
\multicolumn{20}{c}{ } & \multicolumn{3}{c}{$\downarrow$} \\
\end{tabular}

\begin{tabular}{ |c| c| c| c| c | c| c| c| c| c| c |c| c| c| c| c | c| c| c|c | c| c| c|}
\cline{9-11}\cline{13-15}\cline{17-19}\cline{21-23}
\multicolumn{8}{c|}{\textcolor{white}{1}}  && & & &  & &   & & & &  & & & &   \\

\cline{9-11}\cline{13-15}\cline{17-19}\cline{21-23}
\multicolumn{8}{c|}{\textcolor{white}{1}} & {\textcolor{white}{1}} & \textcolor{white}{1} & \textcolor{white}{1} & 
& \textcolor{white}{1}  & &   & & & &  & & & &   \\

\cline{9-11}\cline{13-15}\cline{17-19}\cline{21-23}
\multicolumn{8}{c|}{\textcolor{white}{1}}  & &&& $\leftarrow$ & &&& $\leftarrow$ & &&& $\leftarrow$ & \rlap{/}{5}&& \\

\cline{9-11}\cline{13-15}\cline{17-19}\cline{21-23}

\multicolumn{8}{c|}{\textcolor{white}{1} } &  &{\textcolor{white}{7}} & \rlap{/}{9}& & \rlap{/}{8} & {\textcolor{white}{7}} &  9 & & 8& \rlap{/}{7}& 9 & & 8& 7&9   \\

\cline{9-11}\cline{13-15}\cline{17-19}\cline{21-23}
\multicolumn{4}{c}{\textcolor{white}{111111111111111111111}} & \multicolumn{3}{c}{} & \multicolumn{1}{c}{} & \multicolumn{3}{c}{$t=10$} & 
\multicolumn{1}{c}{ }& \multicolumn{3}{c}{$t=9$} & \multicolumn{1}{c}{} & \multicolumn{3}{c}{$t=8$} &
\multicolumn{1}{c}{ }& \multicolumn{3}{c}{$t=7$} \\
\end{tabular}
\caption{Sequence of moves for the bay of Example \ref{ex.flex.ret}: flexible case ($m=1$)}
\label{ex.flex.ret.solution1}
\end{subfigure}
\caption{ }
\label{temp}
\end{figure}

The first two moves for both cases are relocating container $c_4$ to the first column and retrieving container $c_1$ from slot $[3,2]$. Next, there are two decisions that are made differently in the case of $m=0$ and $m=1$:

$1$. After retrieving $c_1$, container $c_2$, which is covered by $c_3$, is next to be retrieved if a FCFS policy is used. But $c_3$ can be retrieved and loaded onto its truck (because 2 moves have been performed, it is time $3$). In the case of $m=0$, $c_3$ is relocated to the third column. In the case of $m=1$, $c_3$ is retrieved out-of-order (before container $c_2$).

$2$. After retrieving $c_4$, container $c_5$ is next to be retrieved in a FCFS policy. But $c_6$ can be retrieved, because 5 moves (in the case of $m=1$) or 6 moves (in the case of $m=0$) have already been performed and it is time $6$ or $7$. In the case of $m=0$, container $c_6$ is relocated to the third column. In the case of $m=1$, $c_6$ is retrieved out-of-order (before container $c_5$).

The results of these two different decisions are that the number of relocations, total delay, and individual delays (denoted by $w_i$ in Table \ref{ex.flex.rexults}) are less in the case of $m=1$. Note that the number of out-of-order retrievals before each container is at most 1 (although the total number of such retrievals is 2 in this example). In fact, in this example, the trucks that pick up containers $c_1$, $c_3$, $c_4$, $c_6$, $c_7$, $c_8$, and $c_9$ do not experience any out-of-order retrievals. Another important point to note is that even for the trucks that experience 1 out-of-order retrieval (that is, trucks that pick up containers $c_2$ and $c_5$ ), their wait time decreases by 1 and 2 units, respectively.
 \end{example}


Table \ref{ex.flex.rexults} summarizes the results for two levels of retrieval process flexibility. Individual delays ($w_1$ to $w_9$) denote the wait time for trucks picking up containers $c_1$ to $c_9$. The total number of relocations and total delay drop by a significant amount (67\% and 32\%, respectively) when we allow for out-of-order retrieval. Moreover, each truck has a shorter wait time when such a policy is in effect.

\begin{table}[htb!]
\begin{center}
\begin{tabular}{ l c  c  c c  c  c c  c  c c   c}
\hline
 & Total  & Total  & \multicolumn{9}{c}{Individual delay} \\
 & relocations & delay & \quad $w_1$ &  $w_2$ &  $w_3$ &  $w_4$ &  $w_5$ &  $w_6$ &  $w_7$ &  $w_8$ &  $w_9$\\
 
 \hline
 $m=0$ & 3 & 22 & 1	& 2	& 2	& 2	& 3	& 3	& 3	& 3	& 3\\

 $m=1$ & 1 & 9 & 1	& 2	& 0	& 1	& 2	& 0	& 1& 1	& 1\\
 
\hline
 
\end{tabular}
\end{center}
\caption{{\protect {Number of relocations and delay for two levels of flexibility in Example \ref{ex.flex.ret} }}}
\label{ex.flex.rexults}
\end{table}

In general, we can always obtain an improved or the same quality solution by allowing out-of-order retrievals. As shown in proposition \ref{pro1}, given any feasible sequence of moves, we can construct another sequence of moves that has the same or fewer relocations and the same or less retrieval delay, if any out-of-order retrieval is possible.

\begin{proposition} \label{pro1}

Suppose $M$ is a feasible sequence of moves for CRP with $m \geq 0$ out-of-order retrievals. Given $M$, we can construct a new sequence of moves, $M'$, that is feasible for CRP with $m+1$ out-of-order retrievals and,

\begin{description}

  \item[(a)] The total number of relocations in $M'$ is at most the number of relocations in $M$; and

  \item[(b)] The delay of all containers for $M'$ is at most the delay in $M$.

\end{description}

\end{proposition}

\begin{proof}
The procedure for constructing $M'$, similar to that illustrated in Example \eqref{ex.flex.ret}, is to retrieve a container that needs to be relocated but is ready to depart.

Suppose the total number of relocations in $M$ is $R$. For all $1 \le r \le R$, we define $t(r)$ as the time of the $r^{th}$ relocation, and $n(r)$ as the index of the container that is relocated at $t(r)$. Also, we denote the number of out-of-order retrievals that take place before retrieving container $c_n$, by $\sigma(c_n)$.

To construct $M'$, we replace each relocation $1 \le r \le R$ with retrieval, whenever the following two conditions are satisfied: \textbf{i.} $d_{n(r)}\le t(r) $; and \textbf{ii.} $\sigma(c_{i}) < m+1,$  for any container $(c_i,d_i)$  such that $d_i<d_{n(r)}$ and $c_i$ is still in the bay. In other words, if a relocated container is ready to depart (condition \textbf{i}), we retrieve it rather than relocate it, provided that the number of out-of-order retrievals for each container does not exceed $m+1$ (condition \textbf{ii}). 
Suppose relocation $\hat{r}$ satisfies condition \textbf{i} and \textbf{ii}; thus we retrieve $c_{n(\hat{r})}$ at $t(\hat{r})$. For the remaining moves in the sequence (after $t_{\hat{r}}$), we form $M'$ by replacing any moves that involve $n(\hat{r})$ with an idle time-step.

Now we show that $M'$ is a feasible solution with $m+1$ out-of-order retrievals. In terms of retrieval times, the only difference between $M$ and $M'$ is that some containers are retrieved earlier; for such containers, condition \textbf{i} ensures that they are retrieved only after their departure time. For the resulting configuration to remain feasible at each time-step, adjustments might be necessary. For example, suppose that there is a move in $M$ at $t>t(\hat{r})$ to relocate a container to slot $[i,j]$, which happens to be on top of $n({\hat{r}})$. In $M'$, we perform the same move except that we relocate the container to slot $[i,j-1]$, because container $n({\hat{r}})$ in $M'$ is not in the bay after time $t(\hat{r})$. Similarly, some of the retrieval moves in $M'$ will be from one slot lower. Finally, condition \textbf{ii} ensures that the number of out-of-order retrievals for each container does not exceed $m+1$.

Next, we show that statement $(a)$ and $(b)$ hold for $M'$:
If there exist at least one $\hat{r}$ that satisfies \textbf{i} and \textbf{ii}, then the number of relocations is decreased by at least one. Moreover, because all containers in $M'$ are retrieved earlier or at the same time as in $M$, the delay of each container is at most its delay in $M$. 

Finally, using Lemma \eqref{Lemma1}, some of the idle times in $M'$ can possibly be eliminated; thus, the delay of some containers in $M'$ can be reduced below the delay in $M$.

\end{proof}

In the next section, we study computationally the effect of out-of-order retrievals on the number of relocation moves and retrieval delay.

\subsection{Computational Experiments for Flexible Retrieval Planning}
\label{sec 4.1}

To evaluate the out-of-order retrieval policy, we perform numerical experiments using the integer program introduced in Section \ref{formulation}. The main goal of the computational experiments is to study the following questions: What is the quantitative impact of an out-of-order retrieval policy on the number of relocation moves and on the average retrieval delay?; how does the improvement change as we vary bay size?; and what is the impact of an out-of-order retrieval policy on service equity?

\paragraph{Effect of an out-of-order retrieval policy on the number of relocations and on delay:}
To study the impact on the number of relocations and on delay, we solve random instances for five different bay sizes with $4$, $5$, $6$, $7$, and $8$ columns. For each bay size, there are 4 tiers in each column. The initial configuration of each random instance is such that the first 3 tiers are full and there is no container in the top tier; so there are initially $12$, $15$, $18$, $21$ and $24$ containers in the bay, respectively. For each of the 5 bay sizes, we solve the random instances with different levels of flexibility, setting the parameter $m$ in constraint \eqref{flexible ret} to 0, 1, and 2. Changing the bay size and the level of flexibility, there are $15$ different test settings. We solve $1000$ random instances for each test setting and take the average over the instances to compute the number of relocations and average delay (per truck). 

Figure \ref{m1m2} shows the results of our numerical experiments. The bay size (the number of columns in the bay) is shown on the horizontal axis. The percent decrease in the total number of relocations is shown in Figure \ref{ave rel}, and the percent decrease in the average delay per truck is shown in Figure \ref{ave delay}. In both plots, the percent decrease is computed relative to the base ($m=0$, the FCFS policy) for two levels of flexibilitys.

\begin{figure}[htb!]
\begin{subfigure}[t]{0.5\textwidth}
\centering
\includegraphics[width=1.1 \textwidth]{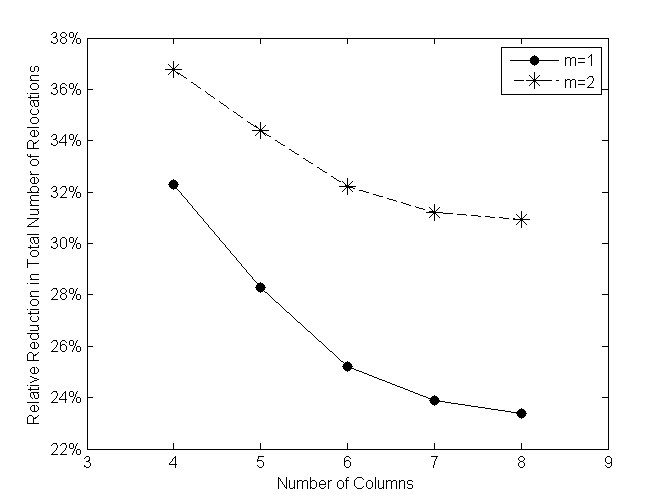}
\caption{ }
\label{ave rel}
\end{subfigure}
\hspace{0.03 in}
\begin{subfigure}[t]{0.5\textwidth}
\includegraphics[width=1.1 \textwidth]{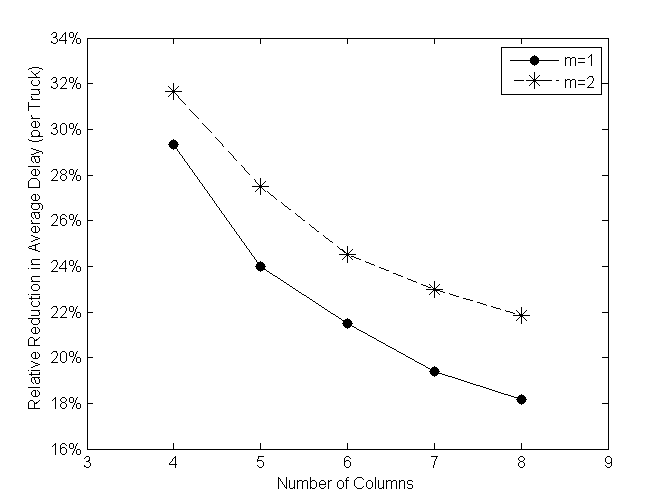}
\caption{ }
\label{ave delay}
\end{subfigure}
\caption{Effect of an out-of-order retrieval policy on the total number of relocations and on average delay}
\label{m1m2}
\end{figure}

As can be seen in Figure \ref{ave rel}, 1 out-of-order retrieval reduces the number of relocation moves by $32\%$ in a bay of size $4\times4$ and by $23\%$ in a bay of size $4\times8$. Note that the decrease in the number of relocations depends on the bay size. The percent decrease becomes smaller as the bay gets larger, but with a decreasing rate. To understand why, note that an out-of-order retrieval of container $(c_n,d_n)$ has two benefits: first, we avoid relocating $c_n$, and thus reduce the total number of relocations by 1. Second, we avoid future relocations that might be incurred by relocating $c_n$ (if we relocate $c_n$ to a ``bad" column that has a container with higher priority, we need to relocate $c_n$ again later). When the bay gets larger, it becomes more likely that we find a ``good" column (a column that has no container with departure time smaller than $d_n$ ) for relocating $c_n$. Therefore, the beneficial effect of out-of-order retrieval diminishes as the bay gets larger.

The dotted line in Figure \ref{ave rel} illustrates that allowing for 2 out-of-order retrievals has a larger effect on the number of relocations. The relative decrease is $37\%$ and $31\%$ for a bay of size $4\times4$ and $4\times8$, respectively.

Figure \ref{ave delay} depicts the impact of an out-of-order retrieval policy on the average delay (per truck). A similar trend can be seen for the relative decrease in the average delay. Allowing for 1 out-of-order retrieval results in a $32\%$ decrease for a $4\times4$ bay and a $22\%$ decrease for a $4\times8$ bay. 

\paragraph{Effect of an out-of-order retrieval policy on service equity:}
 When a FCFS is replaced by an out-of-order retrieval policy, truck drivers might perceive service to be unfair when another truck is served out-of-order before them.
However, a truck can receive its container out-of-order or can experience an out-of-order retrieval before it is served.
To formalize this, we compute the number of out-of-order retrievals performed before serving a particular truck as follows:

\begin{figure}[htb!]
\begin{subfigure}[t]{0.5\textwidth}
\centering
\includegraphics[width=1 \textwidth]{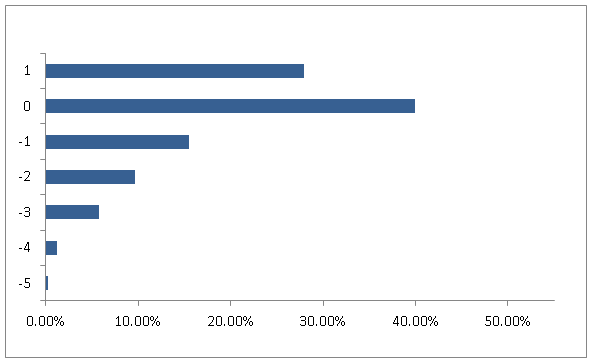}
\caption{$4\times4$ bay; $m=1$}
\label{fig:4_4_1}
\end{subfigure}
\hspace{0.05 in}
\begin{subfigure}[t]{0.5\textwidth}
\includegraphics[width=1 \textwidth]{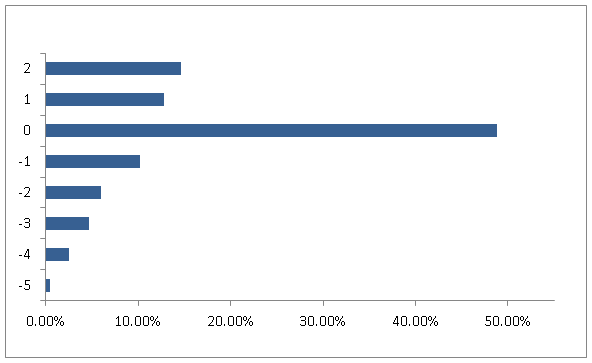}
\caption{$4\times4$ bay; $m=2$}
\label{fig:4_4_2}
\end{subfigure}

\vspace{0.05 in}

\begin{subfigure}[t]{0.5\textwidth}
\centering
\includegraphics[width=1 \textwidth]{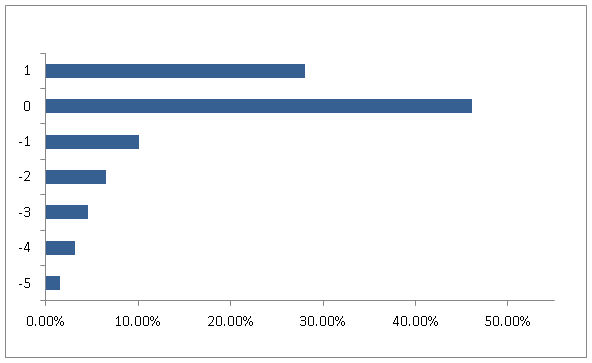}
\caption{$8\times4$ bay; $m=1$}
\label{fig:4_8_1}
\end{subfigure}
\hspace{0.05 in}
\begin{subfigure}[t]{0.5\textwidth}
\includegraphics[width=1 \textwidth]{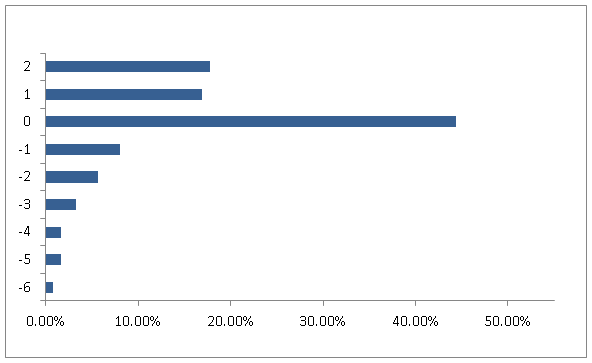}
\caption{$8\times4$ bay; $m=2$}
\label{fig:4_8_2}
\end{subfigure}
\caption{Long term histogram of out-of-order retrieval moves performed before truck $r$ is served: negative numbers indicate that truck $r$ is served out-of-order; positive numbers indicate that other truck(s) are served out-of-order; zero indicates that neither of these cases happen}\label{fig:equity}

\end{figure}

Figure \ref{fig:equity} shows the distribution of the number of out-of-order retrievals for truck $r$, for two bay sizes (a small bay with 4 columns and a large bay with 8 columns), and two levels of flexibility ($m=1$ and $m=2$). The integer numbers on the vertical axis show the number of out-of-order retrieval moves that truck $r$ experiences in the long term and the horizontal axis shows the frequency of each case happening. In a flexible retrieval process, three cases can occur when truck $r$ is served: 

${\bullet}$ Truck $r$ is not served out-of-order, and does not experience any out-of-order retrieval. This case corresponds to number zero on the vertical line.

${\bullet}$ One or more trucks are served out-of-order before truck $r$ (a truck that is served out-of-order and also experiences some out-of-order retrievals is also classified in this group). This case corresponds to the positive numbers on the vertical axis. The positive numbers indicate the number of trucks that are served out-of-order; note that the maximum of these positive numbers depends on $m$. For example in Figures \ref{fig:4_4_1} and \ref{fig:4_8_1}, only $+1$ is on the vertical axis because $m=1$. In Figures \ref{fig:4_4_2} and \ref{fig:4_8_2}, there are two positive values ($+1$ and $+2$) because $m=2$. 

${\bullet}$ Truck $r$ is served out-of-order. For example it is the $n^{th}$ truck in the line, but is served as the $m^{th}$ truck where $m<n$. This case corresponds to negative numbers on the vertical line. The value of the negative numbers indicates how much earlier truck $r$ is served. For example the case that truck $r$ is the $n^{th}$ truck and is served as the $m^{th}$ truck, corresponds to $m-n$. Similarly, the case that truck $r$ is served as the third truck is indicated by value $-2$.

As shown in Figure \ref{fig:equity}, more than 40\% of the time truck $r$ does not experience out-of-order service for another truck and is also not served out-of-order.

Moreover,, the frequency that truck $r$ is served out-of-order (negative numbers) is almost the same as the frequency that it experiences out-of-order retrievals (positive numbers), when $m=1$. These experiments imply that in the long term, the service that each truck receives is not adversely affected by of the out-of-order retrieval policy. Further, as shown in Figure \ref{ave delay}, the average retrieval delay (per truck) is decreased as a result of out-of-order retrievals.

\section{Concluding Remarks}
\label{discussion}
In the storage yard of container terminals, the import containers are retrieved from the bays to be delivered to the external trucks (retrieval process), and, the export containers are stacked in the bays after being discharged from the vessels (stacking process). Service times (departure and arrival times of containers, or some estimate of them) are used by port operators to plan the retrieval and stacking process. In this paper, we introduce a time-based model and integer programming formulation that exploits such information. Unlike existing models, our model incorporates service times and thereby enables port operators to manage repositioning and plan the container movements when containers depart and arrive continuously. Moreover, our model captures customer service elements, such as service delay, and allows for adjusting levels of service.

We also present a flexible retrieval planing policy and studied its impacts on the number of relocations and retrieval delay. We showed that such a policy improves operations by decreasing (or holding unchanged) the number of relocations while the waiting time of external trucks is shortened or remains unaffected. 

There are many papers that have proposed and studied heuristics using order-based approaches for the CRP and DCRP. Similar to other optimization models for the CRP, our model is also not suitable for real-time decision making for large instances of the problem. Thus, efforts to design time-based heuristics is an interesting direction for future work. Moreover, it seems that using service times in the model is more convenient when solving the CRP with incomplete information (for example, when we have time windows as an estimate of departure and arrival times of containers, or some probabilistic distribution on the service times). Therefore, studying the CRP with incomplete information in a time-based framework is another important direction for future research.

\bibliographystyle{apa}
\bibliography{ref20}

\newpage
\section*{Appendix A\\ {\large Planning reshuffling moves: Trade-off between relocation and service delay}}
\label{App_examples}


Consider a $3 \times 3$ bay with 6 containers, with the initial configuration given in Figure \ref{fig:rel.delay}. Let $S = (1,3,4,5,6,7)$ be the arrival times of the trucks that pick up containers $c_1$ to $c_6$. Notice that $t=2$ is an idle time-step because the first container is not blocked and can be retrieved at $t=1$. For this simple example, there are two solutions that result in different values for number of relocation moves and total delay as shown in Table \ref{table:ex1}.

\begin{figure}[htb!]
\centering
\begin{tabular}{ |c| c| c|}
\hline
 & & \\
\hline
$(c_1,1)$&$(c_2,3)$&$(c_5,6)$\\
\hline
$(c_4,5)$&$(c_6,7)$&$(c_3,4)$\\
\hline
\end{tabular}
\caption{Initial configuration of the bay }
\label{fig:rel.delay}
\end{figure}

In the first solution, no move is done at $t=2$, whereas in the second solution, $(c_5,6)$ is relocated during the second time-step, resulting in less delay for $(c_3,4)$ later in the retrieval process. However in the second solution, $(c_5,6)$ needs to be relocated again, resulting in a total of 2 relocations. With the objective function being the weighted sum of the relocation moves and total delay, the optimal solution would depend on the relative weights that determine the trade-off between the two metrics. For example if we assign equal weights to both terms in the objective function, both solutions would be optimal as they have same objective function value. If we assign a larger weight to relocation than delay, the first solution (with 1 relocation and 4 delay) would be optimal.

\begin{table}[htb!]
\begin{center}
\begin{tabular}{ l c  c}
\cline{2-3}
 & Solution 1 & Solution 2 \\
\hline
t=1 & retrieve $(c_1,1)$ : [1,2] $\rightarrow$ out & retrieve $(c_1,1)$ : [1,2] $\rightarrow$ out  \\
\textbf{t=2} & \textbf{idle} & \textbf{relocate  $(c_5,6)$ : [3,2] $\rightarrow$ [1,2]}  \\
t=3 & retrieve $(c_2,3)$ : [2,2] $\rightarrow$ out & retrieve $(c_2,3)$ : [2,2] $\rightarrow$ out  \\
t=4 & relocate  $(c_5,6)$ : [3,2] $\rightarrow$ [2,2] & retrieve $(c_3,4)$ : [3,1] $\rightarrow$ out  \\
t=5 & retrieve $(c_3,4)$ : [3,1] $\rightarrow$ out & relocate $(c_5,6)$ : [1,2] $\rightarrow$ [2,2]  \\
t=6 & retrieve $(c_4,5)$ : [1,1] $\rightarrow$ out & retrieve $(c_4,5)$ : [1,1] $\rightarrow$ out  \\
t=7 & retrieve $(c_5,6)$ : [2,2] $\rightarrow$ out & retrieve $(c_5,6)$ : [2,2] $\rightarrow$ out  \\
t=8 & retrieve $(c_6,7)$ : [2,1] $\rightarrow$ out & retrieve $(c_6,7)$ : [2,1] $\rightarrow$ out  \\
 \hline
 Relocations & 1 & 2\\
 Total delay & 4 & 3 \\
 \hline
\end{tabular}
\end{center}
\caption{Two sequence of moves for the bay shown in Figure \ref{fig:rel.delay}: if $w_{rel} = w_r$, both solutions are optimal}
\label{table:ex1}
\end{table}

\newpage
\section*{Appendix B\\ {\large Proof of Lemma \ref{Lemma1}}}

\begin{proof}
Suppose $M$ is an optimal sequence of moves, and let $M(t)$ be the move of time $t$. Also, suppose there exist a time step $\tilde{t}$, and a container $(c_{\tilde{n}},d_{\tilde{n}})$ such that: \textbf{i.} $M(\tilde{t})$ is idle; and \textbf{ii.} $c_{\tilde{n}}$ is in the bay at time $\tilde{t}$, and $d_{\tilde{n}} \le \tilde{t}$. Given $M$, we can construct another optimal solution, $M'$, as follows:

We swap the two moves, $M(\tilde{t})$ and $M(\tilde{t}+1)$. As a result, $M(\tilde{t}+1)$ becomes idle. We repeat the swapping for $M(\tilde{t}+1)$ and $M(\tilde{t}+2)$, if conditions \textbf{(i)} and  \textbf{(ii)} are satisfied for $M(\tilde{t}+1)$. We continue this process until we reach a time step that does not satisfy the two conditions. In the resulting sequence of moves, $M'$, the idle time step that was initially at $\tilde{t}$ is postponed to a later time step (that does not satisfy conditions \textbf{(i)} and  \textbf{(ii)}). Next we show that $M'$ is feasible and optimal.

First note that an idle time step does not make any changes in the configuration; therefore adding or removing or shifting an idle time does not violate any of the constraints that involve feasibility of configuration. For the same reason, constraints \ref{const4c} that ensure at most one move is performed at each time step, will not be violated. Secondly, note that by postponing the idle time step, the retrieval time of a container in $M'$ can only be earlier than its retrieval time in $M$. Moreover, because we check condition \textbf{(ii)} for each swapping, it is ensured that a container is retrieved only after its departure time. Thus, all retrievals in $M'$ are within the allowable time windows. Lastly, the containers are retrieved in the prescribed order in $M'$, because we only swap an idle time step with its next move, and such a swapping does not affect the initial retrieval order of containers in $M$. 

The resulting sequence of moves, $M'$, is also optimal; because the only difference between $M$ and $M'$ is that the order of some moves are different, which does not affect the total number of relocation moves.
\end{proof}

\newpage
\section*{Appendix C \\ {\large Example for DCRP}}
\label{appendixC}

Consider a set of 9 containers ($c_1,\dots,c_{9}$) that need to be stacked in a $3\times4$ bay, which is initially empty. The stacking schedule is given by $S_s = (1,1,1,1,1,2,2,2,2)$, which implies that the first five containers have to be stacked before the last four containers, but the containers in each group do not have to be stacked in a particular order (this is enforced by Constraint \ref{stack order}). The retrieval process starts after the stacking process (i.e. there is no overlap between the two processes in this example). During the retrieval process, containers $c_1$ to $c_9$ will be picked up by external trucks whose arrival schedule is given by $S_r = (15,19,18,20,16,17,21,22,23)$. The objective is sum of retrieval delay, stacking delay, and total number of relocations. Solution to this problem is the sequence of moves for stacking and later retrieving the containers, as well as the relocation moves. Table \ref{table:ex.DCRP} summarizes the solution.


\begin{table}[htb!]
\begin{center}
\begin{tabular}{ l c c l c}
\cline{2-2}\cline{5-5}
 & Stacking process & \textcolor{white}{11} & & Retrieval process \\
\cline{1-2}\cline{4-5}
t=1 & stack $c_5$ : out  $\rightarrow$ [1,1]  & &  t=15 & retrieve $(c_1,15)$ : [1,2] $\rightarrow$ out  \\
t=2 & stack $c_1$ : out  $\rightarrow$ [1,2]  & & t=16 & retrieve $(c_5,16)$ : [1,1] $\rightarrow$ out  \\
t=3 & stack $c_4$ : out  $\rightarrow$ [2,1] & & t=17 & retrieve $(c_6,17)$ : [3,4] $\rightarrow$ out  \\
t=4 & stack $c_2$ : out  $\rightarrow$ [2,2] & & t=18 & retrieve $(c_3,18)$ : [2,3] $\rightarrow$ out  \\
t=6 & stack $c_{3}$ : out  $\rightarrow$ [2,3] & & t=19 & retrieve $(c_2,19)$ : [2,2] $\rightarrow$ out  \\
t=7 & stack $c_9$ : out  $\rightarrow$ [3,1]  & &  t=20 & retrieve $(c_4,20)$ : [2,1] $\rightarrow$ out  \\
t=8 & stack $c_8$ : out  $\rightarrow$ [3,2] & & t=21 & retrieve $(c_7,21)$ : [3,3] $\rightarrow$ out  \\
t=9 & stack $c_7$ : out  $\rightarrow$ [3,3] & & t=22 & retrieve $(c_8,22)$ : [3,2] $\rightarrow$ out \\
t=10 & stack $c_6$ : out  $\rightarrow$ [3,4] & & t=23 & retrieve $(c_9,23)$ : [3,1] $\rightarrow$ out \\

\cline{1-2}\cline{4-5}
\end{tabular}
\end{center}
\caption{Sequence of stacking and retrieval moves}
\label{table:ex.DCRP}
\end{table}

As shown in Table \ref{table:ex.DCRP}, the optimal sequence of stacking is $c_5,c_1,c_4,c_2,c_3$ for the containers of first group and $c_9,c_8,c_7,c_6$ for containers of the second group. With this stacking plan, the retrieval process has no relocation move and the total retrieval delay is zero. If we had to specify a fixed stacking order as an input to the model, the retrieval delay and relocations could be more because it is not obvious what is the best stacking order (that minimizes the objective). For example, if we set the stacking order as $c_1,c_2,c_3,c_4,c_5,c_{9},c_8,c_7,c_6$, total relocations and total retrieval delay would be 3 and 23, respectively. Notice that as shown in this example, the time-based model for DCRP can also be used for determining the best slot to stack a container given its pick up time in the future.

\newpage
\section*{Appendix D \\ {\large Integer programming formulation for CRP and DCRP}}

\begin{eqnarray}
&&\left.
\begin{array}{l l}
\sum\limits_{i=1}^{C}\sum\limits_{j=1}^{P}b_{ijnt}+v_{nt}= 1,\nonumber\quad n=1,...,N,\quad t=1,...,T  & \quad \textrm{for CRP}  \vspace{0.4cm}\\ \nonumber
\sum\limits_{i=1}^{C}\sum\limits_{j=1}^{P}b_{ijnt}+v_{nt}=z_{nt},\quad n=1,...,N,\quad t=1,...,T & \quad\textrm{for DCRP} \vspace{0.4cm}\\\nonumber
\sum\limits_{n=1}^{N}b_{ijnt} \leq 1,\quad i=1,...,C,\quad j=1,...,P,\quad t=1,...,T &  \vspace{0.4cm}\\\nonumber
\sum\limits_{n=1}^{N}b_{ijnt} \geq \sum\limits_{n=1}^{N}b_{ij+1nt},\quad i=1,...,C,\quad j=1,...,P-1,\quad t=1,...,T  &  \vspace{0.4cm}\\\nonumber
\sum\limits_{i,k=1}^{C}\sum\limits_{j,l=1}^{P}\sum\limits_{n=1}^{N} x_{ijklnt} + \sum\limits_{i=1}^{C}\sum\limits_{j=1}^{P}\sum\limits_{n=1}^{N}y_{ijnt} \leq 1,\quad t=1,...,T 
& \quad  \textrm{for CRP} \vspace{0.4cm}\\\nonumber
\sum\limits_{i,k=1}^{C}\sum\limits_{j,l=1}^{P}\sum\limits_{n=1}^{N} x_{ijklnt} + \sum\limits_{i=1}^{C}\sum\limits_{j=1}^{P}\sum\limits_{n=1}^{N}y_{ijnt} + \sum\limits_{i=1}^{C}\sum\limits_{j=1}^{P}\sum\limits_{n=1}^{N}s_{ijnt} \leq 1,\quad t=1,...,T & \quad   \textrm{for DCRP} \vspace{0.4cm}\\ \nonumber
(t-m-n)\sum\limits_{i=1}^{C}\sum\limits_{j=1}^{P}y_{ijnt} + \sum\limits_{i=1}^{C}\sum\limits_{j=1}^{P}\sum\limits_{n=1}^{N}\sum\limits_{t^\prime=1}^{t-1}y_{ijnt^\prime}\le t-1\\\nonumber
 \quad \quad \quad \quad \quad \quad \quad \quad t=1,\dots,T,\quad n=1,\dots,N \quad \textrm{and} \quad d_1<d_2<\dots<d_N &  \vspace{0.4cm}\\\nonumber
\sum\limits_{t=1}^{T}z_{nt} \geq \sum\limits_{t=1}^{T}z_{n^{\prime}t},\quad \forall n, n^{\prime} \quad s.t. \quad a_{n^{\prime}}=a_n + 1 & \quad  \textrm{for DCRP}  \vspace{0.4cm}\\ \nonumber
b_{ijnt}=b_{ijnt-1}-\sum\limits_{k=1}^{C}\sum\limits_{l=1}^{P} x_{ijklnt-1} + \sum\limits_{k=1}^{C}\sum\limits_{l=1}^{P} x_{klijnt-1} - y_{ijnt-1} &\quad  \textrm{for CRP} \\ \nonumber 
\quad \quad \quad \quad \quad \quad \quad  i=1,...,C,\quad j=1,...,P,\quad n=1,...,N,\quad t=2,...,T &  \vspace{0.4cm}\\\nonumber
b_{ijnt}=b_{ijnt-1}-\sum\limits_{k=1}^{C}\sum\limits_{l=1}^{P} x_{ijklnt-1} + \sum\limits_{k=1}^{C}\sum\limits_{l=1}^{P} x_{klijnt-1} - y_{ijnt-1}+s_{ijnt-1}& \quad \textrm{for DCRP}\nonumber \\
\quad \quad \quad \quad \quad  \quad \quad \quad \quad \quad i=1,...,C,\quad j=1,...,P,\quad    n=1,...,N,\quad t=2,...,T &  \vspace{0.4cm}\nonumber\\
v_{nt} =\sum\limits_{i=1}^{C}\sum\limits_{j=1}^{P}\sum\limits_{t^{'}=1}^{t-1}y_{ijnt^{'}},\quad n=1,...,N,\quad t=1,...,T & \quad  \textrm{for CRP} \vspace{0.4cm}\\\nonumber
z_{nt} =\sum\limits_{i=1}^{C}\sum\limits_{j=1}^{P}\sum\limits_{t^{'}=1}^{t-1}s_{ijnt^{'}},\quad n=1,...,N,\quad t=1,...,T & \quad  \textrm{for DCRP}  \vspace{0.4cm}\\\nonumber
\end{array} \right. \\ \nonumber 
\end{eqnarray}
 
 \begin{eqnarray}
 &&\left\{
 \begin{array}{l}
 \vspace{0.4cm}
 \sum\limits_{i=1}^{C}\sum\limits_{j=1}^{P}\sum\limits_{t=1}^{d_{n}-1}y_{ijnt}=0,\quad \quad \quad \quad n=1,...,N \nonumber\\
 \sum\limits_{i=1}^{C}\sum\limits_{j=1}^{P}\sum\limits_{t=d_{n}}^{d_{n}+\delta_n}y_{ijnt}=1,\quad\quad\quad\textrm{\textcolor{white}{1}} n=1,...,N \nonumber\\
 \sum\limits_{i=1}^{C}\sum\limits_{j=1}^{P}\sum\limits_{t=d_{n}+\delta_n+1}^{T}y_{ijnt}=0,\quad \quad n=1,...,N \vspace{0.4cm}\nonumber\\
 \end{array} \right.  \\ \nonumber 
 &&\left\{
\begin{array}{l}
\vspace{0.4cm}
\sum\limits_{i=1}^{C}\sum\limits_{j=1}^{P}\sum\limits_{t=1}^{a_{n}-1}s_{ijnt}=0,\quad \quad\quad\quad   n=1,...,N \nonumber \\
\sum\limits_{i=1}^{C}\sum\limits_{j=1}^{P}\sum\limits_{t=a_{n}}^{a_{n}+\alpha_n}s_{ijnt}=1,\quad\quad\quad  \textrm{\textcolor{white}{1}} n=1,...,N \nonumber
\quad\quad\quad\quad\quad\quad\quad\quad\quad\quad\quad\quad\quad\quad\quad\quad\quad\textrm{for DCRP} \nonumber \\
\sum\limits_{i=1}^{C}\sum\limits_{j=1}^{P}\sum\limits_{t=a_{n}+\alpha_n+1}^{T}s_{ijnt}=0,\quad\quad   n=1,...,N \vspace{0.4cm}\nonumber \\
\end{array} \right. \nonumber 
\end{eqnarray}

\newpage
\section*{Appendix E \\ {\large Proof of Proposition \ref{m out}}}

\begin{proof}
Proof is by contradiction. 
Suppose there is at least one container with more than $m$ out-of-order retrievals, and let $n$ be one of such containers with $m+1$ out-of-order retrievals. Since there are $m+1$ out-of-order retrievals and we assumed that at most $(n-1)+m$ containers are retrieved before container $(n,d_n)$, we know that the number of containers whose departure time is smaller than $d_n$ and are retrieved before container $(n,d_n)$, is at most $(n-1)+m - (m+1) = n-2$. Therefore, at least one container (from the set $\{(1,d_1),(2,d_2),\dots,(n-1,d_{n-1})\}$) is retrieved after $(n,d_n)$. Without loss of generality, let $(n-1,d_{n-1})$ be the container that is going to be retrieved after $(n,d_n)$. Note that at least $n-1+m$ retrievals (including $m+1$ out-of-order) take place before retrieving container $(n-1,d_{n-1})$, and this contradicts our assumption that before container $n-1$, we have at most $((n-1)-1)+m=n-2+m$ retrievals.

The condition of having at most $(n-1)+m$ retrievals before each container $n$, is enforced by Constraint \eqref{flexible ret}. These constraints ensure that if container $n$ is being retrieved at time $t$, then the number of retrievals performed up to time $t$ does not exceed $(n-1)+m$. Note that the first term in the left hand side is $(t-m-n)\times 1$ whenever container $n$ is being retrieved at time $t$ and therefore the constraint is reduced to $(t-m-n) + \sum\limits_{i=1}^{C}\sum\limits_{j=1}^{P}\sum\limits_{n=1}^{N}\sum\limits_{t^\prime=1}^{t-1}y_{ijnt^\prime}\le t-1$, enforcing the sum of retrievals up to time $t$ to be at most $(n-1)+m$. When container $n$ is not being retrieved at time $t$, the first term in the left hand side is zero and the constraint is reduced to $\sum\limits_{i=1}^{C}\sum\limits_{j=1}^{P}\sum\limits_{n=1}^{N}\sum\limits_{t^\prime=1}^{t-1}y_{ijnt^\prime}\le t-1$ which is simply a loose bound (and a redundant constraint) on the number of retrieval moves up to time $t$.
\end{proof}

\end{document}